\documentclass{article}
\usepackage[noline,noend,linesnumbered,algoruled]{algorithm2e}
\usepackage[normalem]{ulem}
\usepackage{natbib}
\usepackage{amsmath,amsfonts,amsthm,amssymb}
\usepackage{authblk}

\usepackage{tikz}
\usetikzlibrary {positioning}
\definecolor {processblue}{cmyk}{0.96,0,0,0}
\usepackage{framed}
\usepackage{tikz,fullpage}
\usepackage{color}
\usetikzlibrary{arrows,%
                petri,%
                topaths}%
\usepackage{tkz-berge}
\usepackage[position=top]{subfig}
\newtheorem{theorem}{Theorem}
\newtheorem{definition}{Definition}

\newtheorem{proposition}{Proposition}
\newtheorem{lemma}{Lemma}
\newtheorem{claim}{Claim}

\newenvironment{talign}
 {\align}
 {\endalign}
\newenvironment{talign*}
 {\csname align*\endcsname}
 {\endalign}

\usepackage{hyperref}

\title{Fair and Efficient Online Allocations with Normalized Valuations}

\author[a]{Vasilis Gkatzelis} \author[b]{Alexandros Psomas} \author[a]{Xizhi Tan}
\affil[a]{Drexel University, Computer Science}
\affil[b]{Purdue University, Computer Science}

\begin{document}
\date{}
\maketitle

\begin{abstract}
A set of divisible resources becomes available over a sequence of rounds and needs to be allocated immediately and irrevocably. Our goal is to distribute these resources to maximize fairness and efficiency. Achieving any non-trivial guarantees in an adversarial setting is impossible. However, we show that normalizing the agent values, a very common assumption in fair division, allows us to escape this impossibility. Our main result is an online algorithm for the case of two agents that ensures the outcome is envy-free  while guaranteeing $91.6\%$ of the optimal social welfare. We also show that this is near-optimal: there is no envy-free algorithm that guarantees more than $93.3\%$ of the optimal social welfare.
\end{abstract}

\section{Introduction}\label{sec:intro}
We consider a basic problem in online fair division: a set of divisible items become available over a sequence of $T$ rounds (one item per round), and in each round we need to make an irrevocable decision regarding how to distribute the corresponding item among a set $N$ of $n$ agents. The value $v_{it}$ of each agent $i$ for the item in round $t$ is revealed at the beginning of that round and our goal is to ensure that the overall allocation at the end of the $T$ rounds is fair and efficient, despite the information limitations that we face.

Prior work on online resource allocation problems such as the one above has mostly focused on maximizing efficiency. In our setting, this could easily be achieved by fully allocating the item of each round $t$ to the agent $i$ with the largest $v_{it}$ value. However, this approach can often lead to outcomes that are patently unfair, which is unacceptable in many important real-world applications. For example, ensuring that the outcome is fair is crucial for food banks that allocate food each day to soup kitchens and other local charities depending on the demand~\cite{Canice17}, or software engineering companies that distribute shared computational resources among their employees \cite{GBJG20}.

Achieving fairness in such an online setting can be significantly more complicated than just maximizing efficiency. This is mostly due to the fact that reaching a fair outcome may require a more holistic view of the instance at hand. For example, the \emph{fair-share} property (also referred to as proportionality in some contexts), one of the classic notions of fairness, requires that each of the $n$ agents should eventually receive at least a $1/n$ fraction of their total value for all the $T$ items. But, agents who only value highly demanded items are harder to satisfy than agents who value items of low demand, and online algorithms may be unable to distinguish between these two types of agents soon enough. As a result, designing efficient online algorithms that also satisfy the fair-share property is an important, yet non-trivial, task.

In fact, it is easy to show that without imposing any normalization on the agent values, essentially the only algorithm that guarantees the fair-share property is the naive one that equally splits every item among all agents (see Appendix~\ref{app:missing from intro} for a proof). This yields an outcome that is inefficient, unless all agents happen to have the same values. But, the standard approach in fair division is to normalize the agents' values so that they add up to the same constant (that constant is usually $1$). As we show in this paper, this normalization is sufficient for us to escape the strong impossibility result and achieve non-trivial efficiency guarantees while satisfying the fair-share property.  

\subsection{Our results and techniques}
With the exception of a few results in Section~\ref{sec:many-agents}, all of our results focus on instances involving two agents, which already pose several non-trivial obstacles.

We first consider the performance of non-adaptive online algorithms, i.e., algorithms whose allocation decision in each round $t$ depends only on the agents' values for item $t$. A major benefit of these algorithms is that they need not keep track of any additional information, making them easy to implement. We focus on the interesting family of \emph{poly-proportional} algorithms that are parameterized by a value $p\geq 0$, and in each round $t$ allocate to each agent $i$ a fraction of the item equal to $\frac{v^p_{it}}{\sum_{j \in N} v^p_{jt}}$. For $p=0$, we recover the algorithm that splits each item equally among the agents (which satisfies fair-share but can be inefficient), while for $p=\infty$ we get the algorithm that allocates each item to the agent with the highest value (which is efficient but violates fair-share). Another well-studied algorithm from this family, that is used widely in practice, is the \emph{proportional allocation} (or just proportional) algorithm, which corresponds to the case $p=1$.  We show that this algorithm satisfies fair-share and is a significant improvement in terms of efficiency: it guarantees $82.8\%$ of the optimal social welfare (Theorem~\ref{thm: proportional algo}). 

As the value of the parameter $p$ grows, the corresponding poly-proportional algorithm allocates each item more ``agressively'', i.e., a larger fraction goes to the agents with the highest values. As a result, higher values of $p$ lead to increased efficiency, but may also lead to the violation of the fair-share property. We precisely quantify this intuition by first showing that for all $p > 2$ the corresponding poly-proportional algorithm does not satisfy fair-share (Lemma~\ref{lem: poly-prop not envy free}). Then, we show that the poly-proportional algorithm with parameter $p=2$, the \emph{quadratic-proportional} algorithm, satisfies fair-share and guarantees $89.4\%$ of the optimal social welfare (Theorem~\ref{thm: qp}). As a result, we conclude that $89.4\%$ is the optimal approximation achievable by a poly-proportional 
algorithm that satisfies fair-share.

Moving beyond non-adaptive algorithms, we proceed to study the extent to which adaptivity could lead to even better approximation guarantees. With that goal in mind, we propose the family of \emph{guarded poly-proportional} algorithms, which are a slight modification of the poly-proportional algorithm, also parameterized by $p$. We show that every algorithm in this family satisfies fair-share, and our main result is that the guarded poly-proportional algorithm with $p=2.7$ guarantees $91.6\%$ of the optimal social welfare (Theorem~\ref{thm: guarded poly approx}). On the other hand, we prove that no fair-share algorithm (adaptive or non-adaptive) can achieve an approximation to the optimal welfare better than $93.3\%$ (Theorem~\ref{thm : lower bound}), thus establishing that our positive result is near optimal.

To prove our results, we leverage the fact that our algorithms have a closed form expression for the agents' allocations and utilities. We can use this fact and write a mathematical program that computes the worst-case approximation to the optimal welfare over all instances. We use variables $v_t$ for the value of agent $1$ for item $t$ and $\lambda_t$ for the ratio between agents' values. Even though this program is not itself convex (so at first glance it's unclear how useful it is), we show that under a suitable choice of variables and constraints, fixing some of the variables (i.e. treating them as constants) gives a linear program with respect to the remaining variables. The majority of the constraints in this LP are non-negativity constraints, so, using the fundamental theorem of linear programming we conclude that the worst-case instance only has a few (two or three depending on the algorithm) items with positive valuations. Once we have such small instances we can analyze the approximation using simple calculus. See the proofs of Theorem~\ref{thm: proportional algo},~\ref{thm: qp} and~\ref{thm: guarded poly approx} for details.

We conclude with a brief discussion regarding instances with $n\geq 3$ agents. We already know from the work of \cite{CKKK12} on the \emph{price of fairness}  that even offline algorithms cannot achieve an approximation better than $\Omega(1/\sqrt{n})$; we complement this result by showing that the non-adaptive proportional algorithm matches this bound. Finally, we provide an interesting local characterization of all online algorithms that satisfy the fair-share property.

\section{Related Work}
The same model that we consider in this setting, i.e., online allocation of divisible items with normalized agent valuations, was very recently studied by~\citet{GBJG20}. But, rather than introducing fairness as a hard constraint, like we do here, they (approximately) maximize the Nash social welfare objective. On the other hand, \citet{BMS19} maximize efficiency subject to fair-share constraints, like we do, but not in an adversarial setting. The agent values are stochastically generated and fairness is guaranteed only in expectation. 

An additional motivation behind our assumption that the agents' values are normalized comes from systems where the users are asked to express their value using a budget of some artificial currency in the form of tokens. If a user has a high value for a good then she can use more tokens to convey this information to the algorithm. Since all users have the same budget, their values are normalized by design. A natural, and very well-studied algorithm in these systems is the \emph{proportional algorithm}, which distributes each item in proportion to the expressed value (see, e.g., \cite{Zhang05,FLZ09,CST16,BGM17}). We provide an analysis of this algorithm, but we also achieve improved results using alternative algorithms.

\citet{ZP20} considered the trade-off between fairness and efficiency under a variety of adversaries, but in a setting with indivisible items and non-normalized valuations. Against the strong adversary studied here, their results are negative: no algorithm with non-trivial fairness guarantees can Pareto-dominate a uniformly random allocation.

More broadly, our paper is part of the growing literature on online, or dynamic, fair division. Much of this prior work analyzes settings where the agents are static and the resources arrive over time, like we do~\cite{walsh2011online,benade2018make,he2019achieving}. 
Another line of work studies the allocation of static resources among dynamically arriving and departing agents \cite{KPS14,friedman2015dynamic,friedman2017controlled,im2020dynamic}.

\section{Preliminaries}

We consider the problem of allocating $T$ divisible items among a set $N$ of $n$ agents.  A fractional allocation $\mathbf{x}$ defines for each agent $i\in N$ and item $t$ the fraction $x_{it}$ of that item that the agent will receive. A feasible allocation satisfies $\sum_{i \in N} x_{it} \leq 1$ for all items $t$. 

We assume the valuations of the agents are \emph{additive}: each agent $i$ has valuation $v_{it}$ for item $t$, and utility $u_i(\mathbf{x}) = \sum_{t \in [T]} v_{it} x_{it}$ for an allocation $\mathbf{x}$. We also assume that the agents' valuations are normalized so that $\sum_{t \in [T]} v_{it} = 1$. We evaluate the efficiency of an allocation $\mathbf{x}$ using the \emph{social welfare} (SW), i.e., the sum of all agents' utilities $SW(\mathbf{x}) = \sum_{i \in N} u_i(\mathbf{x})$.

An allocation $\mathbf{x}$ satisfies fair-share if $u_i(\mathbf{x}) \geq \frac{1}{n}$ for every agent $i \in N$. We say that an algorithm satisfies fair-share if it always outputs an allocation that satisfies fair-share. Another popular definition of fairness is envy-freeness, which dictates that no agent $i$ values the allocation of some other agent $j$ more than her own. It is well known that if every item $t$ is fully allocated, i.e., $\sum_{i\in N}x_{it}=1$, then envy-freeness implies fair-share, and for two-agent instances (which is the main focus of this paper) the two notions coincide.

The item valuations are not available to us up-front; instead, the items arrive online (one per round) and the agent values for the item of round $t$ are revealed when the item arrives. The
algorithm then makes an irrevocable decision about how to allocate the item before moving on to the next round. We evaluate our algorithms using worst-case analysis, so one can think of the values being chosen by an adaptive adversary aiming to hurt the algorithm's performance. Throughout the paper our algorithms do not need to know the total number of rounds $T$, but all our negative results apply even to algorithms that have this information.

We say an algorithm is \textit{non-adaptive} if its allocation decision for round $t$ solely depends on the valuations at round $t$, whereas an $\textit{adaptive}$ algorithm can use the valuations and allocations of all the previous rounds. An interesting family of non-adaptive algorithms parametrized by a value $p$ are ones that we call \emph{poly-proportional} algorithms whose allocation in each round $t$ is proportional to $v_{it}^p$, i.e., each agent $i$ is allocated a fraction $x_{it} =  v^p_{it}/\sum_{j \in N} v^p_{jt}$. For $p=0$ this become the equal-split algorithm, for $p=1$ the proportional algorithm, and for $p=\infty$ the greedy one. 

Given some algorithm $\mathcal{A}$, let $\mathbf{x}^{\mathcal{A}}(\mathbf{v})$ denote the overall allocation that it outputs on an instance with agent values $\mathbf{v}$, and let $\mathbf{x}^{\text{OPT}}(\mathbf{v})$ be the social welfare maximizing allocation. $\mathcal{A}$ is an $\alpha$-approximation to the optimal social welfare if 
\[
\min_{\mathbf{v}} \frac{SW(\mathbf{x}^{\mathcal{A}}(\mathbf{v}))}{SW(\mathbf{x}^{\text{OPT}}(\mathbf{v}))} \geq \alpha.
\]
Note that our algorithms are constrained to be online and to always output fair-share outcomes, while the welfare maximizing benchmark is restricted by neither one of the two.

\section{Non-Adaptive Algorithms}\label{sec: non adaptive}

Non-adaptive algorithms have the important benefit that they need not keep track of historical information regarding the agents' allocation or preferences. A naive example of such an algorithm is equal-split, i.e,. the poly-proportional algorithm with $p=0$. Since this algorithm splits every item equally among the two agents, they both always receive value exactly $1/2$, and hence the outcome is fair-share. However, this outcome can be very inefficient, leading to a $50\%$ approximation to the optimal welfare (e.g., consider an instance with $v_{11}=v_{22}=1$ and $v_{12}=v_{21}=0$).

Our first result analyzes the widely-used proportional algorithm ($p=1$) and shows that it guarantees $82.8\%$ of the social welfare. This is already a big improvement compared to $50\%$,
but we then also provide a fair-share algorithm that improves this further, to $89.4\%$. Proofs missing from this section can be found in Appendix~\ref{app: missing from non adaptive}.

\begin{theorem}\label{thm: proportional algo}
The proportional algorithm satisfies fair-share and gives a $0.828$ approximation to the optimal welfare.
\end{theorem}
\begin{proof}
First we porve the envy-freeness of the proportional algorithm. We will use Milne's inequality~\cite{milne1925note} which states that for all $x_j, y_j > 0$:

\[
\sum_{j=1}^m \frac{x_j y_j}{x_j + y_j} \leq \frac{(\sum_{j=1}^m x_j)(\sum_{j=1}^m y_j)}{\sum_{j=1}^m x_j + \sum_{j=1}^m y_j}.
\]

Plugging in $x_j = v_{1j}$ and $y_j = v_{2j}$, the LHS is exactly the value of agent $1$ for agent $2$'s allocation, while the RHS is equal to $1/2$.

We now proof the efficiency guarantees of the proportional algorithm. Given an instance $\mathbf{v}$, let $v_t = v_{1t}$ and $\lambda_t = \frac{v_{2t}}{v_{1t}}$ for each $t \in [T]$. Let $ALG$ be the welfare of the proportional algorithm.
\[  ALG =\sum_{t \in [T]} \frac{v_t^2+(v_t\lambda_t)^2}{v_t+ v_t\lambda_t} = \sum_{t \in [T]} v_t \frac{1+\lambda_t^2}{1+\lambda_t}. \]
Now, consider the following mathematical program:
\begin{talign}
\text{minimize}&\sum_{t \in [T]} v_t \frac{1+\lambda_t^2}{1+\lambda_t} \notag \\
	\text{subject to} 
	&\sum_{t \in [T]} v_t = \sum_{t \in [T]} v_t \lambda_t  \label{constr: equal values} \\
	&\sum_{t \in [T] :\lambda_t \leq 1} v_t + \sum_{t \in [T]: \lambda_t > 1} \lambda_t v_t = 1 \label{constr: opt is 1}\\
	&v_t,\lambda_t \geq 0,  \text{ for all } t \in [T] \notag
\end{talign}
The objective is to minimize the approximation to welfare we receive from the algorithm. In this program, we don't enforce that the agents' values add up to $1$, but we simply have them be equal to each other (constraint~\ref{constr: equal values}). Instead, we ask that the optimal welfare is equal to $1$ (constraint~\ref{constr: opt is 1}).

First, we argue that solving this program would give us the worst case approximation to welfare.  Consider an arbitrary feasible solution $\mathbf{v} , \mathbf{\lambda}$ to this program; by dividing each agents' values (each $v_{it}$) by their common total value $\sum_{t \in [T]} v_{it}$ we get a feasible instance for the original problem. Furthermore, the approximation to welfare in this instance is equal to the value of the objective: the social welfare of the proportional algorithm and the optimal social welfare are the program's objective and $1$, divided by the normalization term $\sum_{t \in [T]} v_{it}$, respectively. Showing that an arbitrary online instance gives a feasible solution to this program with the approximation to welfare unchanged is equally straightforward.

Second, notice that for any fixed $\mathbf{\lambda}$, the remaining program, with variables only the $v_t$s, is a linear program with $T$ variables. By the fundamental theorem of linear programming, a minimizer occurs at the region's corner, i.e. there is a minimizer with $T$ constraints tight. Since the total number of constraints is $T+2$, and the first two constraints are tight, $T-2$ of the $T$ tight constraints are non-negativity constraints. So the worst case approximation occurs when there are exactly two variables/rounds with positive value for agent 1. Without loss of generality (the proportional algorithm is memoryless) these are the first two items.

Third, for every instance where agent $1$ values only the first two items, the approximation to optimal welfare is minimized when agent $2$ also values only the first two items. 

Now, consider the two rounds instance, in the original notation, where agent $1$ has value $v_1$ for item $1$ and $1-v_1$ for item $2$, while agent $2$ has values $1-v_2$ and $v_2$. 
Without loss of generality $v_1 \geq 1-v_2$, which implies $v_2 \geq 1- v_1$. Therefore, $OPT = SW(\mathbf{x}^{\text{OPT}}(\mathbf{v})) = v_1 + v_2$, and
\[
ALG = \frac{v_1^2+(1-v_2)^2}{v_1+1-v_2}+\frac{(1-v_1)^2+v_2^2}{v_2+1-v_1}.
\]
Then, overloading notation, we have that the approximation to the welfare is
\[
 \alpha(v_1,v_2) = \frac{\frac{v_1^2+(1-v_2)^2}{v_1+1-v_2}+\frac{(1-v_1)^2+v_2^2}{v_2+1-v_1}}{ v_1 + v_2}.
\]
We analyze this function, by taking partial derivatives and analyzing all critical points. We find that the worst approximation to optimal welfare is achieved for $v_1 = v_2 = 1/\sqrt{2}$, and has value
$\alpha\left(\frac{1}{\sqrt{2}},\frac{1}{\sqrt{2}}\right) = 2(\sqrt{2}-1) \approx 0.828$. See Appendix~\ref{app: missing from non adaptive} for the missing details.
\end{proof}

\subsection{Performance of poly-proportional algorithms}

We now study the family of poly-proportional algorithms more broadly. As we mentioned in the introduction, poly-proportional algorithms with higher values of $p$ may lead to increased social welfare, but they also make it increasingly likely that the fair-share property will be violated. We first show that we cannot increase $p$ by too much before losing fair-share: for any $p > 2$ the corresponding poly-proportional algorithm does not satisfy fair-share.

\begin{lemma}\label{lem: poly-prop not envy free}
The poly-proportional algorithm with parameter $p$ does not satisfy fair-share for any $p>2$.
\end{lemma}

\begin{proof}
Consider the following two item instance. The first round has values $x$ and $1$ for agents $1$ and $2$, respectively, while the second round has values $1-x$ and $0$. Agent $1$ has utility $\frac{x \cdot x^p}{1 + x^p} + 1-x = 1 - \frac{x}{x^p + 1}$.
For $x = (\frac{1}{p-1})^{1/p}$, agent $1$ gets utility $1 - \frac{p-1}{p}(\frac{1}{p-1})^{1/p}$. For all $a > 0$, $b \in (0,1)$, we have that $a^b > \frac{a}{a+b}$, thus the utility  of agent $1$ is
$1 - \frac{p-1}{p}\left(\frac{1}{p-1}\right)^{\frac{1}{p}} < 1 - \frac{p-1}{p} \frac{\frac{1}{p-1}}{\frac{1}{p-1} + \frac{1}{p}}=1 - \frac{p-1}{2p-1}$. This expression is less than $1/2$ for all $p>2$.
\end{proof}

Our main result in this section is for the poly-proportional algorithm with parameter $p=2$: we call this the \emph{quadratic-proportional} algorithm. We show that this algorithm satisfies fair-share and achieves a $0.894$ approximation to the optimal welfare, a significant improvement over the proportional algorithm.  By Lemma~\ref{lem: poly-prop not envy free}, the quadratic-proportional algorithm guarantees the optimal social welfare within the class of fair-share poly-proportional algorithms.

\begin{theorem}\label{thm: qp}
The quadratic-proportional algorithm satisfies fair-share and achieves a $0.894$ approximation to the optimal social welfare.
\end{theorem}

Theorem~\ref{thm: qp} follows from Lemmas~\ref{lem: qp is envy free} and~\ref{lem: qp efficiency}.

\begin{lemma}\label{lem: qp is envy free}
The quadratic-proportional algorithm satisfies fair-share.
\end{lemma}

\begin{proof}
It suffices to show that agent $1$ gets utility at least $1/2$ in all instances: if this holds, then the same holds for agent $2$, by symmetry. Given any instance, we first show that merging and splitting certain items(rounds) results in a new instance where agent $1$ is worse off.

Merging a set $S$ of items with values $(v_{1t},v_{2t})$ creates a new item with value $(\sum_{t \in S} v_{1t} , \sum_{t \in S} v_{2t} )$. A split operation on an item with values $(v_1,v_2)$, $v_1 \geq v_2$, creates two items, with values $(v_2,v_2)$ and $(v_1-v_2,0)$. 

\begin{claim}\label{claim:split}
Let $\mathbf{v}$ be any instance, and let $\mathbf{v}'$ be the instance where we split all items $t \in [T]$ such that $\frac{v_{1t}}{v_{2t}} > 1$, with $v_{2t} > 0$. Then the utility of agent $1$, in the quadratic-proportional algorithm, in instance $\mathbf{v}'$ is at most her utility in instance $\mathbf{v}$. 
\end{claim}

\begin{proof}
It suffices to show that the utility of agent $1$ weakly decreases after splitting a single item with values $v_1 = x, v_2 =y$, such that $\frac{x}{y} \geq 1$. Let $u$ be the utility of agent $1$ (for this item) before splitting and $u^*$ the utility after splitting. We have that $u = \frac{x^{p+1}}{x^p+y^p}$ and $u^* = \frac{y^{p+1}}{2y^p} + x - y = x - \frac{y}{2}$.
\[
u - u^* = \frac{y^{p+1} - 2xy^p + y x^p}{2x^p + 2y^p}.
\]

It suffices to show that this is non-negative for all $x \geq y$. Since $2x^p + 2y^p \geq 0$, we only need to show that $y^{p+1} - 2xy^p + y x^p \geq 0$. Dividing both sides by $y^{p+1}$, we have $1  - 2 \frac{x}{y} + (\frac{x}{y})^p \geq 0$. For $p=2$, the LHS is equal to $(x/y - 1)^2$ which is non-negative. Note that we used the fact that $x > y$ to ensure that splitting was a valid operation. 
\end{proof}

\begin{claim}\label{claim:merge}
Let $\mathbf{v}$ be any instance, and let $\mathbf{v}'$ be the instance where we take two arbitrary items of $\mathbf{v}$ that satisfy $\frac{v_{1t}}{v_{2t}} \leq 1$ and merge them. Then the utility of agent $1$, in the quadratic-proportional algorithm, in instance $\mathbf{v}'$ is at most her utility in instance $\mathbf{v}$. 
\end{claim}

\begin{proof}
Let $a$ and $b$ be the two items we want to merge, with corresponding values $v_{1a}, v_{2a}, v_{1b}$ and $v_{2b}$. We show that
\[
\frac{v_{1a}^3}{v_{1a}^2+v_{2a}^2} + \frac{v_{1b}^3}{v_{1b}^2 + v_{2b}^2} \geq \frac{(v_{1a}+v_{1b})^3}{(v_{1a}+v_{1b})^2 + (v_{2a}+v_{2b})^2}.
\]

We can simplify this expression to:
\begin{align*}
(v_{2b} v_{1a} - v_{2a} v_{1b})^2 & \left( v_{2b}^2 v_{1a} + 2 v_{2b} v_{2a} (v_{1a} + v_{1b}) \right. \\
& \left.  + v_{1b} (v_{2a}^2 - v_{1a} (v_{1a} + v_{1b})) \right) \geq 0.
\end{align*}

If $v_{2b} v_{1a} - v_{2a} v_{1b} = 0$ we are done. Assume that this is not the case. It suffices to show that 
\[
v_{2b}^2 v_{1a} + 2 v_{2b} v_{2a} (v_{1a} + v_{1b}) + v_{1b} (v_{2a}^2 - v_{1a}^2 - v_{1a}v_{1b}) \geq 0.
\]

First, we are going to drop the second term of the sum. Second, since $\frac{v_{1a}}{v_{2a}} \leq 1$, we have that $v_{2a}^2 \geq v_{1a}^2$, and the third term is lower bounded by $v_{1a} v_{1b}^2$. It thus remains to show that $v_{2b}^2 v_{1a} - v_{1a} v_{1b}^2 \geq 0$, which holds since $\frac{v_{1b}}{v_{2b}} \leq 1$.
\end{proof}

We repeatedly apply Claims~\ref{claim:split} and \ref{claim:merge}, until no splitting or merging is possible, to get a worst case instance for agent $1$. This instance will have multiple items with zero value for agent $2$ that we can simply combine into a single item. Since splitting is no longer possible, there are no items $t \in [T]$ with $v_{2t} > 0$ and $\frac{v_{1t}}{v_{2t}} > 1$. Since merging is not possible there is at most one item $t$ with $\frac{v_{1t}}{v_{2t}} \leq 1$. Therefore, we have an instance with two items, one with both positive values (that we cannot merge) and one with zero value for agent $2$. Let $v$ be the value of agent $1$ for item $1$, and $1-v$ her value for item $2$. Agent $2$'s values are $1$ and $0$.

Agent $1$ has utility $\frac{v^3}{v^2 + 1} + 1-v = 1 - \frac{v}{v^2 + 1}$. It is easy to confirm that this function is minimized for $v=1$ where it takes the value $1/2$.
\end{proof}

\begin{lemma}\label{lem: qp efficiency}
The quadratic-proportional algorithm 
achieves a $0.894$ approximation to the optimal social welfare.
\end{lemma}

We start by showing that two item instances are the worst case. This is, in fact, true for all algorithms in the poly-proportional family.

\begin{claim}\label{claim: poly prop worst case}
For any $p$, the worst-case instance (in terms of approximation) for the poly-proportional algorithm with parameter $p$ has at most two items.
\end{claim}

\begin{proof}
Similarly to Theorem~\ref{thm: proportional algo} one can write a mathematical program with variables $v_t$ and $\lambda_t$ that computes the worst case approximation to welfare, and then observe that for every fixed choice of $\mathbf{\lambda}$ the remaining program is in fact linear. Applying the fundamental theorem of linear programming we conclude that at most two $v_t$ variables are non-zero. We defer the details to Appendix~\ref{app: qp efficiency }.
\end{proof}

\begin{proof}[Proof of Lemma~\ref{lem: qp efficiency}]
Given Claim~\ref{claim: poly prop worst case} we only need to consider two item instances. Let $v_1$ and $1-v_2$ be the agents' values for item $1$, and $1-v_1$ and $v_2$ their values for item $2$.

Without loss of generality, assume that $v_1 > 1-v_2$ (and therefore $v_2 > 1-v_1$). The optimal welfare becomes $OPT = v_1 + v_2$. Consider the performance of our algorithm:
\[\textstyle
ALG = \frac{v_1^3 + (1-v_2)^3}{v_1^2 + (1-v_2)^2} + \frac{(1-v_1)^3+v_2^3}{(1-v_1)^2+v_2^2}.
\]
The approximation to welfare is
\[
\alpha(v_1,v_2) = \frac{ALG}{OPT} = \frac{\frac{(1-v_1)^3+v_2^3}{(1-v_1)^2+v_2^2} + \frac{(1-v_2)^3+v_1^3}{(1-v_2)^2+v_1^2}}{v_1 + v_2}
\]

\noindent In the remainder of the proof we take partial derivatives with respect to $v_1$ and $v_2$ and analyze the critical points, using numerical solvers for part of the proof. The worst extreme point is $(0.6265, 0.6265)$, which gives $\alpha(0.6265, 0.6265) > 0.894$. See Appendix~\ref{app: qp efficiency } for details.
\end{proof}

\section{Adaptive Algorithms}\label{sec: adaptive}
Moving beyond non-adaptive algorithms, in this section we consider the benefits of being adaptive. In deciding how to allocate the item of each round $t$, adaptive algorithms can take into consideration, e.g., the utility of each agent so far, or what portion of their total value is yet to be realized. But, what would be a useful way to leverage this information in order to achieve improved approximation guarantees?

We propose a natural way to modify the family of poly-proportional mechanisms studied in the previous section. Specifically, we use the additional information to ``guard'' against the violation of the fair-share property. To motivate this modification, assume that at the end of some round $c$ during the execution of a poly-proportional with $p>2$ the utility that one of the agents has received so far plus her value for all remaining items is exactly $1/2$, i.e.,

\[
\sum_{t = 1}^{c} v_{it} x_{it} + \sum_{t = c+1}^T v_{it} = \frac{1}{2}.
\]
This would mean that, unless that agent receives all of the remaining items that she has positive value for in full, then she would not receive her fair share. We refer to this as a \emph{critical point} and use it to define the family of \emph{guarded poly-proportional} algorithms parametrized by $p$: while no agent has reached a critical point, the algorithm is identical to the corresponding non-adaptive poly-proportional one; but, if some agent reaches a critical point, then all the remaining items are fully allocated to that agent. It therefore leverages adaptivity in a simple way, by checking for critical points.

Note that a critical point may not necessarily arise only at the beginning or the end of a round. However, it is easy to show that we can assume this is the case without loss of generality. Roughly speaking, if a critical point is reached during the execution of some round $t$ while a fraction $f$ of that item has been allocated, then we can divide that item into two pieces (of size $f$ and $1-f$), creating an instance with $T+1$ items where the critical point is reached at the end of round $t$, and without affecting the outcome of the algorithm. We discuss this in more detail in Appendix~\ref{app: missing from adaptive}.

If some agent reaches a critical point then, clearly, these algorithms ensure that the agent will receive her fair share. But, this does not imply that the other agent will also receive her fair share. For this to be true, the other should have received her fair share before that critical point, because she will receive no more items.

Our next result shows that, in fact, this family of algorithms always satisfies fair-share.

\begin{lemma}\label{lem: guarded is ef}
The guarded poly-proportional algorithm with parameter $p$ satisfies fair-share for all $p\geq 0$.
\end{lemma}

\begin{proof}
If there is no critical point the statement trivially holds, so assume, without loss of generality, that agent $1$ reaches a critical point at round $c$. By definition, we have that $\sum_{t = 1}^{c} v_{1t} \cdot x_{1t} + \sum_{t = c+1}^T v_{it} = 1/2$. By the normalization assumption, 
$\sum_{t = 1}^{c} v_{1t} \cdot (x_{1t}+x_{2t}) + \sum_{t = c+1}^T v_{it} = 1$. We get
$\sum_{t = 1}^{c} v_{1t} \cdot x_{2t}= \sum_{t = 1}^{c} v_{1t}\cdot \frac{v_{2t}^p}{v_{1t}^p+v_{2t}^p} = \frac{1}{2}$. That is, it remains to show that fair-share is satisfied for agent $2$.

Similarly to the proof of Theorem~\ref{thm: proportional algo} and Lemma ~\ref{lem: qp is envy free} we will write a mathematical program with variables $v_t = v_{1t}$ and $\lambda_t = \frac{v_{2t}}{v_{1t}}$, for all $t \in [t]$. The goal of the program this time will be to find a worst-case instance with respect to agent $2$, given that $c$ is a critical point for agent $1$.

Agent $1$'s utility of resources allocated to agent $2$ can be expressed as $\sum_{t \leq c} v_t \frac{\lambda_t^{p}}{1+ \lambda_t^p}$, while agent $2$ has utility $\sum_{t \leq c} v_t \frac{\lambda_t^{p+1}}{1+ \lambda_t^p}$. Consider the program
 \begin{equation}
	\begin{array}{lcl}
	\text{minimize}&\sum_{t \leq c} v_t \frac{\lambda_t^{p+1}}{1+\lambda_t^p}\\
	\text{subject to}
	&\sum_{t=1}^c v_t \frac{\lambda_t^p}{1+\lambda_t^p} = \frac{1}{2} \\
	&\sum_{t = 1}^c v_t \leq 1 \\
	&\sum_{t = 1}^c v_t \lambda_t \leq 1 \\
	&v_t, \lambda_t \geq 0 \text{ for all } t \in [c]\\
	\end{array}
	\notag
\end{equation}

Notice that given a feasible solution to this program one can always construct a valid online allocation instance, where the guarded poly-proportional algorithm with parameter $p$ will reach critical point $c$ for agent $1$ and agent $2$'s utility is exactly the objective function, and vice versa. Proving the lemma is therefore equivalent to showing that the optimal solution $u_2$ of this program above is at least $\frac{1}{2}$.

Consider any fixed choice for the $\lambda_t$ variables: the remaining program is linear, and therefore, by the fundamental theorem of linear programming we know that there exists an optimal solution with $c$ tight constraints (since there are $c$ variables). The first constraint is already tight, so we have $c-1$ other tight constraints. At least $c-3$ of those are non-negativity constraints, so we have at most 3 positive variables. In the remainder of the proof we consider all the cases; details are deferred to Appendix~\ref{app: missing from adaptive}.
\end{proof}

For non-adaptive algorithms we observed that efficiency increases with $p$ but, unfortunately, the largest value that yields a fair-share algorithm is $p=2$. For the guarded poly-proportional family we can get a fair-share algorithm for all $p$, but how does the efficiency depend on this value? For larger values of $p$, the algorithm is trying to maximize social welfare more aggressively, but this means that it is more likely to reach a critical point, after which it is forced to be inefficient.

Based on a class of instances provided in Appendix~\ref{app: missing from adaptive}, Figure~\ref{fig:cpplot} provides approximation upper bounds quantifying precisely this trade-off: if for each $p$ we restrict our attention to instances where the corresponding poly-proportional algorithm does not reach a critical point, then the performance increases with $p$. But, as $p$ increases, the set of instances with a critical point keeps growing and the greediness of the algorithm gradually hurts its efficiency. 

For each value of $p$ the points in the plot upper bound the algorithm's approximation, so the most promising choice is $p=2.7$, where the two points meet. Our main result is that the guarded poly-proportional with parameter $p=2.7$ achieves a $0.916$ approximation to the optimal social welfare which, as the figure indicates, is essentially optimal within the family of guarded poly-proportional algorithms.

\begin{figure}
    \centering
    \includegraphics[width=11cm]{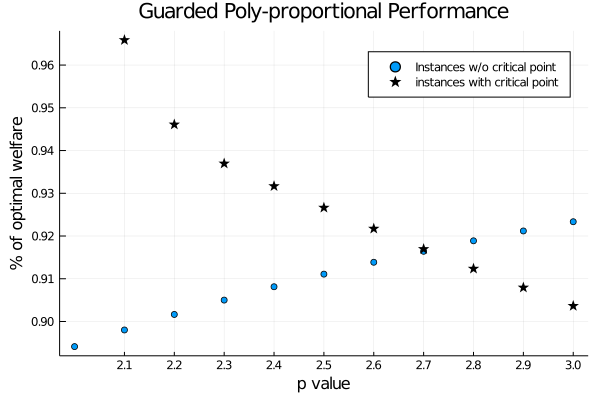}
    \caption{Approximation to the optimal welfare by guarded poly-proportional algorithms for different values of $p$, depending on whether the instance has a critical point or not}
    \label{fig:cpplot}
\end{figure}

\begin{theorem}\label{thm: guarded poly approx}
The guarded poly-proportional algorithm with parameter $p =2.7$ achieves a $0.916$ approximation to the optimal social welfare.
\end{theorem}

\begin{proof}
Let $\alpha$ be the approximation to the optimal welfare of the algorithm. We encode an instance with variables $v_t = v_{1t}$, and $\lambda_t = \frac{v_{2t}}{v_{1t}}$, for all $t \in [T]$. 
Let $c$ be the critical point (if any) and without loss of generality, assume that agent $1$ reaches her critical point.
Agent $2$'s utility $\sum_{t \leq c} v_t \lambda_t \cdot \frac{(v_t \lambda_t)^p}{ v_t^p + (v_t \lambda_t)^p} = \sum_{t \leq c} v_t \frac{\lambda_t^{p+1}}{1+ \lambda_t^p}$. Agent $1$'s utility is $\sum_{t \leq c} v_t \frac{1}{1 + \lambda_t^p} + \sum_{t = c+1}^T v_t$. 
Similarly to Theorem~\ref{thm: proportional algo} and Lemma~\ref{lem: qp efficiency} we write a mathematical program for the optimal approximation ratio:
 \begin{equation}
	\begin{array}{lcl}
	\text{minimize}&\sum_{t = 1}^c v_t \frac{1+\lambda_t^{p+1}}{1+\lambda_t^p} +\sum_{t=c+1}^T v_t \\
	\text{subject to}
	&\sum_{t=1}^T v_t  = 2 ( \sum_{t=1}^c v_t\frac{1}{1+\lambda_t^p} + \sum_{t = c+1}^T v_t ) \\
	&\sum_{t = 1}^T v_t = \sum_{t = 1}^T v_t \lambda_t \\
	&\sum_{t \in [T]:\lambda_t \leq 1} v_t + \sum_{t \in [T]: \lambda_t > 1} \lambda_t v_t =1\\
	&v_t \geq 0, \text{ for all } t \in [T]\\
	&\lambda_t \geq 0, \text{ for all } t \in [T]\\
	\end{array}
	\notag
	\end{equation}

The first constraint encodes the fact that $c$ is a critical point: the LHS is the total value of agent $1$, while the RHS is twice the utility of agent $1$. These should be equal since $c$ is a critical point for agent $1$. The second constraint equalizes the agents' values (instead of normalizing them to $1$), while the third constraint normalizes the optimal welfare to $1$. 
One can go from an arbitrary feasible solution of this program to a valid instance by dividing each $v_{it}$ by $\sum_{t=1}^T v_{t}$, and vice versa, while the approximation to the optimal welfare (which is equal to the welfare when the optimal welfare is $1$) is exactly the objective of this program.

Now observe that for every fixed choice of the $\lambda_t$ variables we get a linear program (with respect to the $v_t$ variables):
 \begin{talign*}
	\text{minimize}&\sum_{t=1}^c v_t a_t+\sum_{t = c+1}^T v_t\\
	\text{subject to}
	&\sum_{t = 1}^T v_t = 2(\sum_{t = 1}^c v_t b_t + \sum_{t=c+1}^T v_t) \\
	&\sum_{t = 1}^T v_t =\sum_{t = 1}^T v_t \lambda_t \\
	&\sum_{t \in [T]:\lambda_t \leq 1} v_t + \sum_{t \in [T]: \lambda_t > 1} \lambda_t v_t =1\\
	&v_t \geq 0, \text{ for all } t \in [T]
\end{talign*}
where $a_t =\frac{1+(\lambda_t)^{p+1}}{1+(\lambda_t)^p}$ and $b_t =\frac{1}{1+(\lambda_t)^p}$.

By the fundamental theorem of linear programming we must have $T$ tight constraints, and we have $T+3$ total constraints (with the first three being tight), so any optimal solution should have exactly $3$ strictly positive $v_t$ variables.

We take cases depending on the value of $c$. Specifically, our three strictly positive $v_t$ variables are either all three after the critical point, two and one, one and two, or all three before the critical point. The first case is, of course, impossible (since the first constraint cannot be satisfied), so we consider each of the other ones.

For each of the cases considered we write a closed form for the approximation to the welfare, as a function of the $\lambda_t$s, we then minimize. For $c=1$ (one item before, two items after the critical point) we get a worst-case approximation of $0.916$. $c=3$, corresponding to no critical points, also gives a worst-case approximation. This corresponds to the intuition from Figure~\ref{fig:cpplot}.
Details can be found in Appendix~\ref{app: missing from adaptive}.
\end{proof}

We complement our positive result by showing that no fair-share adaptive algorithm, even with full knowledge of the number of items $T$, can achieve an approximation to the welfare much better than the guarded poly-proportional family.

\begin{theorem}\label{thm : lower bound}
There is no fair-share algorithm that achieves an approximation to the optimal welfare better than $0.933$.
\end{theorem}
\begin{proof}
Assume that there exists an online algorithm $\mathcal{A}$ that achieves an approximation better than $0.933$, and consider the following two instances. In the first instance, the agents values are $v_{11} = 0.568$ and $v_{21} = 0.427$ in the first round and $v_{12} = 1 - v_{11} = 0.432$ and $v_{22} = 1-v_{21} = 0.573$ in the second round. In the second instance, the agent values, $\mathbf{v'}$, are again, $v'_{11} = 0.568$ and $v'_{21} = 0.427$ in the first round, but their values in the second round are $v'_{12} = 1-v'_{11} = 0.432$ and $v'_{22} = 0.306$ and agent 2's remaining value of $v'_{23}=1-v'_{21}-v'_{22}=0.267$ is realized in the third round. In what follows, we show that no online algorithm can simultaneously satisfy the fair-share property and guarantee an approximation better than 0.933 in both of these two instances. This argument takes advantage of the fact that prior to the second round, no online algorithm can distinguish between these two instances.

\paragraph{Case 1.}
Assume the algorithm allocates less than $0.6977$ of item 1 to agent 1 in the first round, i.e., $x_{11} < 0.6977$, and consider instance 1. Fair-share for agent 1 implies that
\begin{equation*}
    v_{11}x_{11}+v_{12}x_{12} \geq 1/2 ~~\Rightarrow~~
    x_{12} >0.241. 
\end{equation*}
The algorithm's welfare is therefore
\begin{align*}
    &v_{11}x_{11}+v_{21}x_{21}+v_{12}x_{12}+v_{22}x_{22}\\
    &1+(v_{11}-v_{21})x_{11}+(v_{21}-v_{11})x_{12}<1.064,
\end{align*}
while the optimal welfare is $v_{11}+ v_{22} = 1.141$, so
\[\alpha = \frac{ALG}{OPT} < 0.933.\]

\paragraph{Case 2.}
Now, let $x_{11} \geq 0.6977$ and assume that $x'_{22} < \frac{0.427x_{11}-0.194}{0.306}$ is the amount of the item that algorithm $\mathcal{A}$ would allocate to agent 2 in round 2 if the second instance values were realized. In this case, the fair-share property will be violated for agent 2 because her utility is
\begin{align*}
       u_2 &= v_{21}(1-x_{11})+v'_{22}x'_{22}+ 1-v_{21}-v'_{22} \\
        &< 0.427(1-x_{11})+0.306\frac{0.427x_{11}-0.194}{0.306}+0.267\\
        &< 0.5.
\end{align*}

\paragraph{Case 3.} Finally, if $x_{11} \geq 0.6977$ and $x_{22} \geq \frac{0.427x_{11}-0.194}{0.306}$ and
consider the second instance. The agents' utilities are
\begin{align*}
    u_1 &= v_{11}x_{11}+v'_{12}(1-x'_{22})\\
        & = 0.432 +0.568x_{11}-0.432x'_{22} ~~~\text{, and}\\
    u_2 &= v'_{21}(1-x_{11})+v'_{22}x'_{22}+v'_{23}\\
        &= 0.694+0.306x'_{22}-0.427x_{11}.
\end{align*}
This leads to a social welfare of
\begin{align*}
    u_1 + u_2 &= 1.126 + 0.141x_{11}-0.126x'_{22}\\
        & \leq 1.126+0.141x_{11}-0.126\frac{0.427x_{11}-0.194}{0.306}\\
        &\leq 1.206 -0.035x_{11}\\
        &\leq 1.182,
\end{align*}
while the optimal welfare $v'_{11}+v'_{12}+v'_{23} = 1.267$, so
\[\alpha = \frac{ALG}{OPT} <0.933. \qedhere\]
\end{proof}

\section{Instances Involving Multiple Agents}\label{sec:many-agents}
We now briefly turn to instances with $n\geq 3$. \citet{CKKK12} prove that even if we knew all the values in advance, the \emph{price of fairness}, i.e., the worst-case ratio of the optimal social welfare of a fair-share outcome over the social welfare of the optimal outcome, is $O(1/\sqrt{n})$. Our next result shows that the proportional algorithm matches this bound in an online manner, and therefore achieves the optimal approximation. 

\begin{theorem}\label{thm:prop_mul}
The proportional algorithm guarantees a $\frac{1}{2\sqrt{n}}$, i.e., $\Omega(1/\sqrt{n})$, approximation to the optimal social welfare.
\end{theorem}
\begin{proof}
Consider any round $t$ and let $v_{\max}=\max_{i\in N} v_{it}$ be the highest value in this round, and $i^*\in  \arg\max_{i\in N} v_{it}$ be an agent with this value. Let $H$ be the set of agents with $v_{it}\geq v_{\max}/\sqrt{n}$ and $L$ be the set of all the remaining agents. If the portion of the item that the proportional algorithm allocates to the agents in $H$ is at least half of all the item, then the social welfare in this round is at least $v_{\max}/(2\sqrt{n})$. 

On the other hand, if the agents in $L$ are allocated more than half of the item, this means that $\sum_{i\in L} v_{it}/\sum_{i\in N} v_{it}>1/2$. But, $v_{\max}>\sqrt{n} v_{it}$ for all $i\in L$ and thus
\[
v_{\max}> \frac{\sqrt{n}}{|L|}\sum_{i\in L}v_{it} \Rightarrow \frac{v_{\max}}{\sum_{i\in N} v_{it}}> \frac{1}{\sqrt{n}}\cdot\frac{\sum_{i\in L}v_{it}}{\sum_{i\in N} v_{it}},
\]
which implies that $\frac{v_{\max}}{\sum_{i\in N} v_{it}}> \frac{1}{2\sqrt{n}}$,
so the allocation to agent $i^*$ is at least $1/(2\sqrt{n})$, and thus in this case as well the social welfare is at least $v_{\max}/(2\sqrt{n})$.

Since the optimal welfare in $t$ is $v_{\max}$ and the proportional algorithm guarantees a welfare of at least $v_{\max}/(2\sqrt{n})$, summing over all rounds concludes the proof.
\end{proof}

The next result shows that even if we were to restrict the benchmark to be the optimal social welfare subject to the fair-share constraint, still, no online algorithm could achieve an approximation better than $\Omega(1/\sqrt{n})$.
Therefore the proportional algorithm is also optimal with respect to the \emph{competitive ratio} measure, which quantifies the worst case loss of welfare due to the online aspect of the problem alone.

\begin{theorem}\label{thm:onlineefchar}
No online fair-share algorithm can achieve a $\frac{3\sqrt{n}}{n+\sqrt{n}-1}$ approximation to the optimal offline fair-share algorithm. That is, the best feasible approximation is $O(\frac{1}{\sqrt{n}})$.
\end{theorem}

\begin{proof}
Consider an instance with $n$ agents and $\sqrt{n}+n$ rounds. 
In the first $\sqrt{n}$ rounds, for the first $\sqrt{n}$ agent, we have $v_{ii} = \frac{n-1}{n}$, and $v_{it} = 0$, $t \neq i$. For the remaining $n-\sqrt{n}$ agents, we have $v_{jt} = \frac{n-1}{n\cdot\sqrt{n}}$, for all $j > \sqrt{n}$. Then, in the last $n$ rounds, we have $v_{ii} = \frac{1}{n}$, and $v_{it} = 0$ elsewhere, for all $i \in N$. 

In the offline problem, each agent gets $\frac{1}{n}$ from the last $n$ rounds. Therefore the optimal offline fair-share welfare is

\[OPT = \sqrt{n} \cdot \frac{n-1}{n} + n \cdot \frac{1}{n} = \frac{n-1}{\sqrt{n}}+1.\]

We now focus our attention on round $\sqrt{n}$. Note that each agent has remaining value $1/n$ at this round. An online fair-share algorithm needs to plan for the event that the remaining values are all realized in the next round, $\sqrt{n}+1$. In order to satisfy fair-share in this scenario, each agent must have utility at least $\frac{n-1}{n}\frac{1}{n} = \frac{n-1}{n^2}$ at the end of round $\sqrt{n}$.

Consider an agent $i$ with $i > \sqrt{n}$. Since her value for all the items before round $\sqrt{n}$ is $v_{it} = \frac{n-1}{n\sqrt{n}}$, to give this agent utility at least $\frac{n-1}{n^2}$ her total allocation
must be $\sum_{t=1}^{\sqrt{n}} x_{it} \geq \frac{n-1}{n^2}\frac{n\sqrt{n}}{n-1}= \frac{1}{\sqrt{n}}$. This is true for all $i > \sqrt{n}$, so there is $\sqrt{n} - (n-\sqrt{n})\frac{1}{\sqrt{n}} = 1$ of the resources, in the first $\sqrt{n}$, to be allocated among the first $\sqrt{n}$ agents. No matter how this is split, the contribution to the welfare is the same. Let $U^t$ be the social welfare at the end of round $t$.  We have

\[
U^{\sqrt{n}}= 1\cdot\frac{n-1}{n}+(n-\sqrt{n})\frac{n-1}{n^2} = 2-\frac{2\sqrt{n}+n+1}{n\sqrt{n}}.
\]
For the last $n$ rounds our algorithm can make an optimal choice: $ALG = U^{\sqrt{n}}+n\cdot\frac{1}{n} = 3-\frac{2\sqrt{n}+n+1}{n\sqrt{n}}<3$. Therefore, we have $\alpha = \frac{ALG}{OPT}<\frac{3}{\frac{n-1}{\sqrt{n}}+1} = \frac{3\sqrt{n}}{n+\sqrt{n}-1}$.

\end{proof}

\subsection{Characterization of fair-share algorithms}
Our final result provides an interesting characterization of fair-share algorithms that could enable the design of novel algorithms in this setting. This characterization uses a very simple condition, which we refer to as doomsday compatibility, and we show that this myopic condition is necessary, but also sufficient, for guaranteeing that the final outcome will satisfy fair-share.

\begin{definition}[Doomsday Compatibility]\label{def:DC}
We say an allocation $\mathbf{x}^{t} = \{x_{it}\}_{i \in N}$ at day $t$ is \emph{doomsday compatible} if there exists some allocation $\mathbf{x}^{t+1}$ that would make the overall outcome satisfy fair-share, if $t+1$ was the last round, i.e., if all the agents' remaining value was realized in round $t+1$.
\end{definition}

\begin{proposition}
An online algorithm satisfies the fair-share property if and only if its allocation in every round $t$ is doomsday compatible.
\end{proposition}
\begin{proof}
First, it is easy to show that doomsday compatibility in every round $t$ is sufficient for an online algorithm to satisfy fair-share. If this condition is satisfied for all $t$, then it is also satisfied for $t=T-1$ and $t=T$, and thus the final outcome is guaranteed to satisfy fair-share.

Now, we show that this condition is also necessary for the algorithm to satisfy fair-share. Assume that there exists a round $t$ such that the online algorithm's allocation in this round is not doomsday compatible. Then, clearly this algorithm would not be fair-share for the instance where $t+1$ is indeed the last round, i.e., where all of the agents' remaining value is realized in round $t+1$.
\end{proof}

\begin{theorem}
If an algorithm is doomsday compatible in some round $t<T$, then there always exists an allocation $\mathbf{x}^{t+1}$ such that it is also doomsday compatible in round $t+1$.
\end{theorem}
\begin{proof}
Consider any round $t$ where the algorithm's allocation is doomsday compatible. This means that there exists some allocation $\mathbf{\tilde{x}}$ that would achieve fair-share if $t+1$ was the last round. In order to show that we can always maintain doomsday compatibility in round $t+1$, it suffices to show that there always exists some allocation $\mathbf{x}^{t+1}$ for that round and an allocation $\mathbf{x}^{t+2}$ for the next round such that the algorithm would satisfy fair-share if $t+2$ were the last round. We show that, in fact, using $\mathbf{\tilde{x}}$ for both rounds $t+1$ and $t+2$ would satisfy this condition. 

To verify this fact, let $\bar{\mathbf{v}}_{it}$ be the remaining value for each agent $i$ after round $t$, and let $u_{i}$ be the total utility each agent received up to round $t$. Since $\mathbf{\tilde x}$ would make the outcome fair-share if $t+1$ was the last round, for any agent $i$ we have $u_i +\bar{\mathbf{v}}_{it}\mathbf{\tilde{x}} \geq \frac{1}{n}$.
Now, if on the other hand $t+2$ was the last round, let $\mathbf{x}^{t+1} = \mathbf{\tilde{x}}$ and $\mathbf{x}^{t+2} = \mathbf{\tilde{x}}$. Then, for any agent $i$ we would have
\begin{align*}
    &u_i + v_{i(t+1)}\mathbf{x}^{t+1} + (\bar{\mathbf{v}}_{it}-v_{i(t+1)})\mathbf{x}^{t+2}\\
    =& u_i + v_{i(t+1)}\mathbf{\tilde{x}}+(\bar{\mathbf{v}}_{it}-v_{i(t+1)})\mathbf{\tilde{x}}\\
    =&u_i + \bar{\mathbf{v}}_{it}\mathbf{\tilde{x}}\geq \frac{1}{n}.
\end{align*}
Therefore, for $\mathbf{x}^{t+1} = \mathbf{\tilde{x}}$, there exists a $\mathbf{x}^{t+2} = \mathbf{\tilde{x}}$ such that the algorithm is doomsday compatible in round $t+1$.
\end{proof}

\section*{Acknowledgments}

This work was done in part while Alexandros Psomas was visiting the Simons Institute for the Theory of Computing. Work was done in part while Alexandros Psomas was at Google Research, MTV. This work was partially supported by  NSF grant CCF-1755955.

\bibliographystyle{plainnat}
\bibliography{reference}

\appendix

\newpage
\section{Limitations without Normalization}\label{app:missing from intro}

Here, we observe that if values are not normalized, the only fair-share algorithm is equal-split. Consider any algorithm $\mathcal{A}$ that does not always split equally. Let $r$ be the first round in which there exists an agent $i$ that gets $x^{\mathcal{A}}_{ir}<\frac{1}{n}$. Since $r$ is the first such round, we have  $x^{\mathcal{A}}_{jt} = \frac{1}{n}$ for all $j \in N$ and $t < r$.
Therefore,

 \[\sum_{t = 1}^{r}v_{it}x^{\mathcal{A}}_{it} <\frac{1}{n} \sum_{t = 1}^{r}v_{it}.\]
 
 But then, if for all subsequent rounds $k>r$ all agents have zero value, i.e., $v_{ik} = 0$ for all $i\in N$ (or, alternatively, if round $r$ was the last round), algorithm $\mathcal{A}$ would fail to satisfy fair-share for agent $i$.
 
\section{Proofs missing from Section~\ref{sec: non adaptive} }\label{app: missing from non adaptive}

\subsection*{Missing from Theorem~\ref{thm: proportional algo}}

\paragraph{Analysis of $\alpha(v_1,v_2)$.}
Recall that

\begin{align*}
 \alpha(v_1,v_2) &= \frac{\frac{v_1^2+(1-v_2)^2}{v_1+1-v_2}+\frac{(1-v_1)^2+v_2^2}{v_2+1-v_1}}{ v_1 + v_2}\\
 &= \frac{2(1+2v_1v_2-v_1-v_2)}{(1-(v_1-v_2)^2)(v_1+v_2)}
\end{align*}

Taking a partial derivative with respect to $v_1$ we have:

\[
\frac{\partial}{\partial v_1} \alpha(v_1,v_2) = \frac{2 f(v_1,v_2)}{((1-(v_1-v_2)^2)(v_1+v_2))^2},
\]

where $f(v_1,v_2) =  v_1^3 (4v_2-2)+v_1^2 (3-2v_2^2-2v_2) + 2v_1 (v_2^2-v_2) - 2v_2^4+2v_2^3+v_2^2-1$. Furthermore, $\frac{\partial}{\partial v_2} \alpha(v_1,v_2) = \frac{2 f(v_2,v_1)}{((1-(v_1-v_2)^2)(v_1+v_2))^2}$. Therefore, finding all the critical points is equivalent to finding all $v_1, v_2$ such that $f(v_1,v_2) = f(v_2,v_1) = 0$. Let $v_1 = x$ and $v_2 = y$, we have\\

\includegraphics[width = \linewidth]{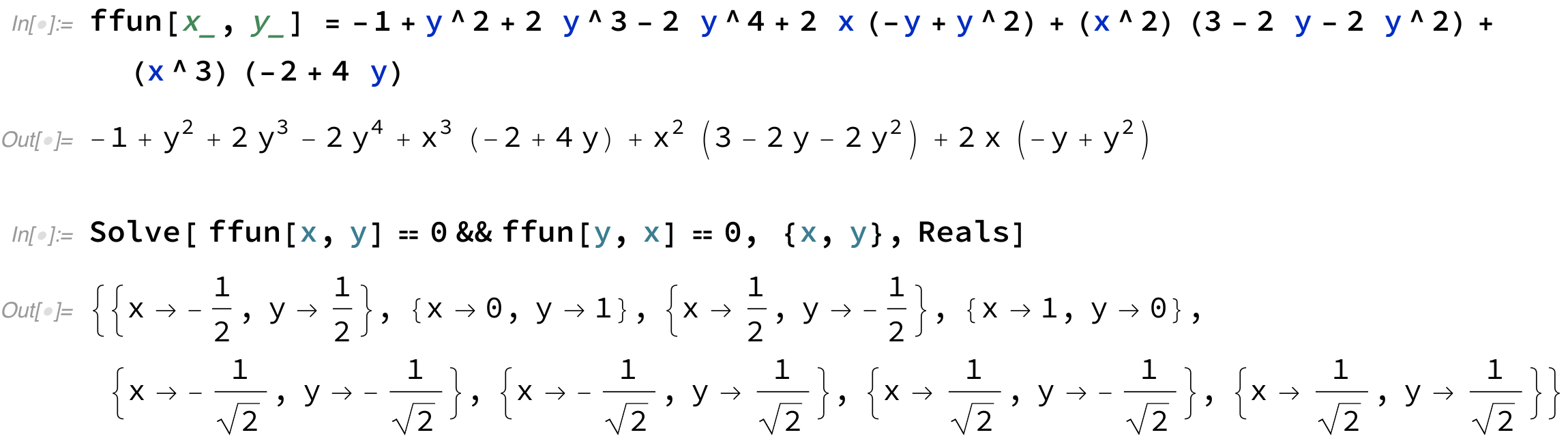}
\\

If $v_1$ or $v_2$ is negative, the corresponding solution is outside of our domain. Furthermore, both $v_1$ and $v_2$ have to be strictly positive (since this is a two item instance), thus we only need to consider one solution $v_1 = v_2 = 1/\sqrt{2}$.

Going back to $\alpha$, the worst approximation to optimal welfare is achieved for $v_1 = v_2 = 1/\sqrt{2}$, and has value
\[
\alpha\left(\frac{1}{\sqrt{2}},\frac{1}{\sqrt{2}}\right) = 2(\sqrt{2}-1) \approx 0.828 \qedhere
\]

\subsection*{Missing from the proof of Claim~\ref{claim: poly prop worst case}}

Given an encoding of an instance as $v_t = v_{1t}$ and $\lambda_t = \frac{v_{1t}}{v_{2t}}$, $ALG$, the social welfare of the quadratic-proportional algorithm can be written as
\[ALG =\sum_{t \in [T]} \frac{v_t^3+(v_t\lambda_t)^3}{v_t^2+ (v_t \lambda_t)^2} = \sum_{t \in [T]} v_t \frac{1+\lambda_t^3}{1+\lambda_t^2}. \]

Now, consider the following mathematical program:
\begin{align*}
\text{minimize}&\sum_{t \in [T]} v_t \frac{1+\lambda_t^3}{1+\lambda_t^2} \notag \\
	\text{subject to} 
	&\sum_{t \in [T]} v_t = \sum_{t \in [T]} v_t \lambda_t   \\
	&\sum_{t \in [T] :\lambda_t \leq 1} v_t + \sum_{t \in [T]: \lambda_t> 1} \lambda_t v_t =1 \\
	&v_t \geq 0,  \text{ for all } t \in [T] \notag \\
	&\lambda_t \geq 0, \text{ for all } t \in [T] \notag
\end{align*}

The objective is to minimize the approximation to the optimal welfare of the algorithm. 
Solving this program would give us the worst case approximation to optimal welfare: consider an arbitrary feasible solution $\mathbf{v} , \mathbf{\lambda}$ to this program. By dividing each agents' values (each $v_{it}$) by their common total value $\sum_{t \in [T]} v_{it}$ we get a feasible instance for the original problem. Furthermore, the approximation to optimal welfare in this instance is equal to the value of the objective: the social welfare of the quadratic-proportional algorithm and the optimal social welfare are the program's objective and $1$, divided by the normalization term $\sum_{t \in [T]} v_{it}$, respectively. Showing that an arbitrary online instance gives a feasible solution to this program with the same approximation is equally straightforward.

Second, notice that for any fixed $\mathbf{\lambda}$, the remaining program, with variables only the $v_t$s, is a linear program with $T+2$ constraints and $T$ variables. By the fundamental theorem of linear program, any minimizer occurs at the region's corner, i.e. any minimizer would have $T$ constraints tight. Since the total number of constraints is $T+2$, at least $T-2$ of the $T$ tight constraints are non-negativity constraints. One can easily observe that the case of exactly $T-1$ tight non-negativity constraints gives an approximation to optimal welfare of $1$ (since there is only one item with $v_{1t} > 0$), so the worst case approximation to optimal welfare happens when there are exactly two positive variables/items with positive value for agent $1$; without loss of generality (poly-proportional allocation algorithms are memoryless) these are the first two items.

Third, for every instance where agent $1$ values only the first two items, the approximation to optimal welfare is minimized when agent $2$ also values only the first two items. 

\subsection*{Missing from the proof of Lemma~\ref{lem: qp efficiency} }\label{app: qp efficiency }

Taking the partial derivative of $\alpha(v_1,v_2)$ with respect to $v_1$ and $v_2$ (and notice that the function is symmetric with respect to $v_1$ and $v_2$) we get that:
\begin{align*}
   \frac{\partial \alpha(v_1,v_2)}{ \partial v_1}  = \frac{f(v_1,v_2)}{(v_1^2+(1-v_2)^2)^2(v_2^2+(1-v_1)^2)^2(v_1+v_2)^2}   \\
   \frac{\partial \alpha(v_1,v_2)}{ \partial v_2} = \frac{f(v_2,v_1)}{(v_2^2+(1-v_1)^2)^2(v_1^2+(1-v_2)^2)^2(v_1+v_2)^2}
\end{align*}
where
\begin{align*}
    &f(v_1,v_2) = -v_1^8 +v_1^7(6-4v_2) +v_1^6(30v_2-18v_2^2-19)\\
    &+v_1^5(-16v_2^3+70v_2^2-82v_2+36)+v_1^4(70v_2^3-16v_2^4 \\
    &-131v_2^2+116v_2-41) +v_1^3(-4v_2^5+42v_2^4-116v_2^3 \\
    &+152v_2^2-104v_2+30) +v_1^2(2v_2^6+2v_2^5-37v_2^4+96v_2^3 \\
    &-116v_2^2+70v_2-17) +v_1(8v_2^7-22v_2^6+14v_2^5+20v_2^4 \\
    &-52v_2^3+58v_2^2-34v_2+8) +v_2^8-6v_2^7+11v_2^6-4v_2^5\\
    &-11v_2^4+18v_2^3-15v_2^2+8v_2-2.
\end{align*}
Let $ v_1=x$ and $ v_2=y$. We write $f(v_1,v_2)$ in Mathematica.
\\

\includegraphics[width=\linewidth]{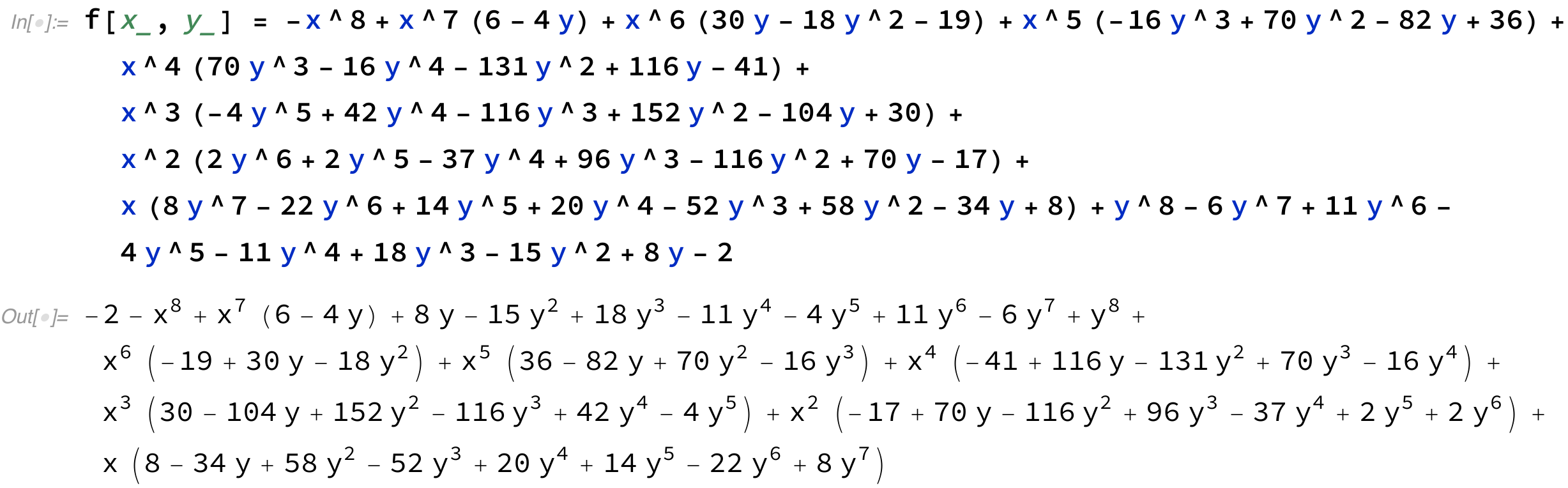}
\\

Solving $\frac{\partial \alpha(v_1,v_2)}{\partial v_1} = \frac{\partial \alpha(v_1,v_2)}{\partial v_2} = 0$ is equivalent to solving $f(v_1,v_2) = f(v_2,v_1) = 0$.
\\

\includegraphics[width=\linewidth]{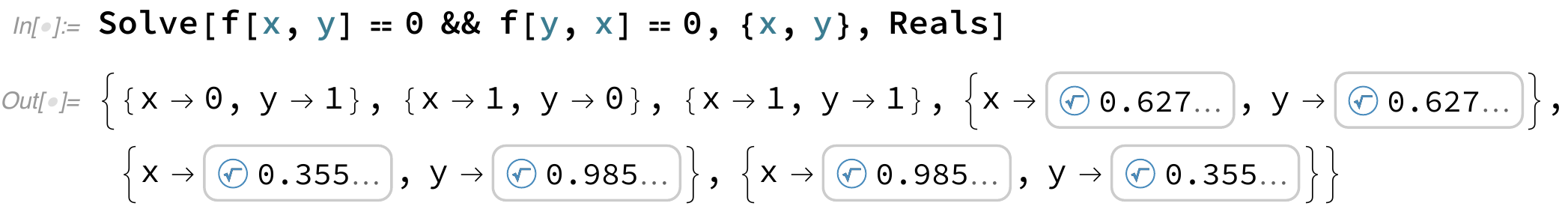}
\\

We numerically confirm that the following are the only solutions: $(0,1)$, $(0.626538,0.626538)$, $(0.35526, 0.985127)$ and $(1,1)$. Notice that $\frac{\partial \alpha(v_1,v_2)}{\partial v_1}(0,1)$ and $\frac{\partial \alpha(v_1,v_2)}{\partial v_2}(0,1)$ are not defined.

Plugging back in the definition of $\alpha$ we have
\begin{gather*}
    \alpha(0,1) = \frac{\frac{1+1}{1+1}}{1} = 1,\\
    \alpha(0.6265, 0.6265) = \frac{2\frac{0.3735^3+0.6265^3}{0.3735^2+0.6265^2}}{1.253} = 0.8941, \\
    \alpha(0.355,0.985) = \frac{\frac{0.645^3+0.985^3}{0.645^2+0.985^2}+\frac{0.015^3+0.355^3}{0.015^2+0.355^2}}{0.355+0.985} = 0.9234, \\
    \alpha(1,1) = \frac{\frac{1}{1}+\frac{1}{1}}{1+1} = 1.
\end{gather*}

We conclude that the algorithm achieves $0.89$ of optimal welfare in the worst case.

\section{Proofs missing from Section~\ref{sec: adaptive} }\label{app: missing from adaptive}

\subsection*{Critical points within rounds}

In this subsection we generalize the notion of a critical point, we discuss how guarded poly-proportional algorithms handle critical points, and why we can, without loss of generality, assume that these points arise only at the beginning (or end) of a round.

Let $\mathbf{v}$ be an instance with $T$ items such that the guarded poly-proportional algorithm with parameter $p$ reaches a critical point for, say, agent $1$ (without loss of generality) \emph{within} some round $c$. This would mean that at the beginning of this round a critical point has not yet been reached, yet at the end of this round it is already too late. Formally,

\[
\sum_{t = 1}^{c-1} v_{1t} x_{1t} + \sum_{t = c}^T v_{1t} ~>~ \frac{1}{2} ~>~ \sum_{t = 1}^{c} v_{1t} x_{1t} + \sum_{t = c+1}^T v_{1t},
\]

where $x_{1t}$ is the fraction of item $t$ allocated using the poly-proportional algorithm with parameter $p$. In this case, we simulate the algorithm as a continuous process, so the critical point is reached after a fraction $f$ of the item of round $c$ has been allocated. Specifically, this is the point when the utility of the agent up to that point, i.e., $\sum_{t=1}^{c-1} v_{1t} x_{1t} +f \frac{(v_{1c})^{p+1}}{(v_{1c})^{p}+(v_{2c})^{p}}$, added to the agent's remaining value, i.e., $(1-f) v_{1c} + \sum_{t=c+1}^T v_{1t}$, adds up to exactly $1/2$. When this point is reached, the algorithm allocates all of the remaining $1-f$ fraction of the item, as well as all subsequent items, to agent 1.

Now, to verify that we can assume this point is always reached at the end of a round, we construct an alternative instance with $T+1$ (instead of $T$) rounds, such that the critical point is reached at the end of round $c$ and the outcome of the algorithm is exactly the same. In fact, all that we need to do is just replace the item of round $c$ with two items of value $(f v_{1c}, f v_{2c})$ for the first one and $((1-f) v_{1c}, (1-f) v_{2c})$ for the second; all other values remain the same as the original instance. The modified instance remains valid, since the agents' values still add up to one and it is easy to observe that the outcome of the algorithm would be the same, while the critical point is reached at the end of round $c$.

\subsection*{Missing from the proof of Lemma~\ref{lem: guarded is ef} }

\begin{proof}[Proof of Lemma~\ref{lem: guarded is ef} continued]

Here, we do the case analysis.

\paragraph{Case 1.} There is exactly one positive variable, say $v_1$, and the second and third constraints are not tight. That is, we have $v_1 < 1$ and $v_1 \lambda_1 < 1$. Since $v_1 < 1$, then the first constraint implies that $\frac{\lambda_1^p}{1+\lambda_1^p} > \frac{1}{2}$, which gives $\lambda_1 > 1$. Thus
\[
\frac{1}{2} =^{\text{(First Const.)}} v_1 \frac{\lambda_1^p}{1+\lambda_1^p} < v_1 \frac{\lambda_1^{p+1}}{1+\lambda_1^p} = \text{Objective}.
\]

\paragraph{Case 2.1.} There are exactly two positive variables, $v_1$ and $v_2$, and the second constraint is tight (i.e. $v_1 + v_2 = 1$), while the third constraint is not tight (i.e. $v_1 \lambda_1 + v_2 \lambda_2 < 1$).

First, we observe that tightness of the second constraint implies that it can't be that both $\lambda_1$ and $\lambda_2$ are strictly bigger than $1$, otherwise we would have $v_1\lambda_1 + v_2\lambda_2 >1$. Furthermore, by the first constraint we cannot have both $\lambda_1$ and $\lambda_2$ strictly less than $1$. 
Thus, assume without loss of generality that $\lambda_1  >1$ and $\lambda_2 <1$. $v_1 + v_2 = 1$ and $v_1 \frac{\lambda_1^p}{1 + \lambda_1^p} + v_2 \frac{\lambda_2^p}{1 + \lambda_2^p} = 1/2$ imply that
\[
v_1 = \left( \frac{1}{2}-\frac{\lambda_2^p}{1+\lambda_2^p} \right) /\left( \frac{\lambda_1^p}{1+\lambda_1^p}-\frac{\lambda_2^p}{1+\lambda_2^p} \right).
\]

\noindent Re-writing the objective we have that the utility of agent $2$ is
\begin{align*}
	    u_2 &= v_1\frac{\lambda_1^{p+1}}{1+\lambda_1^p} + v_2\frac{\lambda_2^{p+1}}{1+\lambda_2^p} \\
	    &= \frac{\lambda_2^{p+1}}{1+\lambda_2^p}+v_1\left(\frac{\lambda_1^{p+1}}{1+\lambda_1^p}-\frac{\lambda_2^{p+1}}{1+\lambda_2^p}\right)\\
	    & = \frac{\lambda_2^{p+1}}{1+\lambda_2^p}+\frac{\frac{1}{2}-\frac{\lambda_2^p}{1+\lambda_2^p}}{\frac{\lambda_1^p}{1+\lambda_1^p}-\frac{\lambda_2^p}{1+\lambda_2^p}}\left(\frac{\lambda_1^{p+1}}{1+\lambda_1^p}-\frac{\lambda_2^{p+1}}{1+\lambda_2^p}\right)\\
	    &= \frac{\lambda_1^{p+1}-\lambda_2^{p+1}+(\lambda_1\lambda_2)^p(\lambda_2-\lambda_1)}{2(\lambda_1^p-\lambda_2^p)}
	\end{align*}

Taking the partial derivative with respect to both $\lambda_1$ we have:
\[\frac{\partial u_2}{\partial \lambda_1} = \frac{f(\lambda_1,\lambda_2)(\lambda_1^p-(\lambda_1\lambda_2)^p)}{2\lambda_1(\lambda_1^p-\lambda_2^p)^2},
\]
where $f(\lambda_1,\lambda_2) = \lambda_1^{p+1}-(p+1)\lambda_1\lambda_2^p+p\lambda_2^{p+1}$. 
Notice that $\frac{\partial f}{\partial \lambda_1} = (p+1)\lambda_1^p-(p+1)\lambda_2^p > 0$, for all $\lambda_1 > \lambda_2$. Therefore $f(\lambda_1,\lambda_2)> f(\lambda_2,\lambda_2) = 0$. Consequently, we have $\frac{\partial u_2}{\partial \lambda_1} >0$ for all $\lambda_1 >1$. In other words, the utility of agent $2$ is lower bounded by $u_2(1,\lambda_2)$, or
\[
u_2(1,\lambda_2) = \frac{1-\lambda_2^{p+1}+\lambda_2^p(\lambda_2-1)}{2(1-\lambda_2^p)} = \frac{1}{2}.
\]

\paragraph{Case 2.2.} 

There are exactly two positive variables, $v_1$ and $v_2$, and the third constraint is tight (i.e. $v_1\lambda_1 + v_2\lambda_2 = 1$), while the second constraint is not tight (i.e. $v_1 + v_2 < 1$).

It's easy to see that we can't have $\lambda_1, \lambda_2 < 1$. If $\lambda_1, \lambda_2 > 1$ we have
\[
v_1 \frac{\lambda_1^{p+1}}{1+\lambda_1^p}+ v_2 \frac{\lambda_2^{p+1}}{1+\lambda_2^p}> v_1 \frac{\lambda_1^{p}}{1+\lambda_1^p}+v_2 \frac{\lambda_2^{p}}{1+\lambda_2^p} = \frac{1}{2},
\]
which is a contradiction (because the first constraint is tight).
Therefore, $\lambda_1 < 1$ and $\lambda_2 > 1$ without loss of generality.

We have
\begin{align*}
    \frac{1}{2} &=^{\text{(First constr.)}} v_1 \frac{\lambda_1^{p}}{1+\lambda_1^p}+v_2 \frac{\lambda_2^{p}}{1+\lambda_2^p} \\
    &=^{\text{(Third constr.)}} v_1 \frac{\lambda_1^{p}}{1+\lambda_1^p}+ \frac{1-v_1\lambda_1}{\lambda_2} \frac{\lambda_2^{p}}{1+\lambda_2^p},
\end{align*}

which implies that
\[
v_1 =\frac{\left(\lambda_2^p+1-2\lambda_2^{p-1}\right)\left(\lambda_1^p+1\right)}{2\left(\lambda_1^p\left(\lambda_2^p+1\right)-\lambda_1\lambda_2^{p-1}\left(\lambda_1^p+1\right)\right)}.
\]

Plugging into the third constraint and re-arranging we get
\begin{align*}
v_2 = \frac{ 2\lambda_1^p\lambda_2^p + 2\lambda_1^p - \lambda_1^{p+1} \lambda_2^p  - \lambda_1^{p+1}  -
\lambda_1 \lambda_2^p  - \lambda_1  }{2\lambda_2 \left(\lambda_1^p\left(\lambda_2^p+1\right)-\lambda_1\lambda_2^{p-1}\left(\lambda_1^p+1\right)\right)}.
\end{align*}

Re-writing the objective we have
\begin{align*}
    &u_2 = v_1\frac{\lambda_1^{p+1}}{1+\lambda_1^p} + v_2\frac{\lambda_2^{p+1}}{1+\lambda_2^p} \\
    &= \frac{ \lambda_1^{p+1}\lambda_2^p+ \lambda_1^{p+1} - 2\lambda_1^{p+1} \lambda_2^{p-1}}{2\left(\lambda_1^p\left(\lambda_2^p+1\right)-\lambda_1\lambda_2^{p-1}\left(\lambda_1^p+1\right)\right)} \\
    &+ \frac{ 2\lambda_1^p\lambda_2^p + 2\lambda_1^p - \lambda_1^{p+1} \lambda_2^p  - \lambda_1^{p+1}  -
\lambda_1 \lambda_2^p  - \lambda_1  }{2 \left(\lambda_1^p\left(\lambda_2^p+1\right)-\lambda_1\lambda_2^{p-1}\left(\lambda_1^p+1\right)\right)} \frac{\lambda_2^{p}}{1+\lambda_2^p} \\
&= \frac{ \lambda_1^{p+1}\lambda_2^p+ \lambda_1^{p+1} - 2\lambda_1^{p+1} \lambda_2^{p-1} - 2\lambda_1^{p+1} \lambda_2^{2p-1}
}{2\left(\lambda_1^p\left(\lambda_2^p+1\right)-\lambda_1\lambda_2^{p-1}\left(\lambda_1^p+1\right)\right)}\frac{1}{1+\lambda_2^p} \\
    &+ \frac{ 2\lambda_1^p\lambda_2^{2p} + 2\lambda_1^p\lambda_2^{p}  -
\lambda_1 \lambda_2^{2p}  - \lambda_1\lambda_2^{p}  }{2 \left(\lambda_1^p\left(\lambda_2^p+1\right)-\lambda_1\lambda_2^{p-1}\left(\lambda_1^p+1\right)\right)} \frac{1}{1+\lambda_2^p} \\
&=  \frac{ \lambda_1^{p + 1} (1 - 2 \lambda_2^{p-1}) + (2 \lambda_1^p - \lambda_1) \lambda_2^{p}
}{2\left(\lambda_1^p\left(\lambda_2^p+1\right)-\lambda_1\lambda_2^{p-1}\left(\lambda_1^p+1\right)\right)}
\end{align*}

Taking the partial derivative with respect to $\lambda_2$ we have
\[\frac{\partial u_2}{\partial \lambda_2} =  \frac{ \lambda_2^p g(\lambda_1,p) f(\lambda_2,\lambda_1)}{2(-\lambda_1\lambda_2^{p-1}(\lambda_1^p+1)+\lambda_1^p(\lambda_2^p+1))^2}\]
where
\begin{align*}
    g(\lambda_1,p) &=\lambda_1-2\lambda_1^p+\lambda_1^{p+1}, \text{ and }\\
    f(\lambda_2,\lambda_1) &= (p-1)\lambda_1^{p+1}+\lambda_1\lambda_2^p -p\lambda_1^p\lambda_2.
\end{align*}

Notice that $\frac{\partial f}{\partial \lambda_2} = p\lambda_1 - p\lambda_1^{p}$, which is positive since $\lambda_1<1$. Thus, $f(\lambda_2,\lambda_1) \geq f(0,\lambda_1) = \lambda_1^{p+1}(p-1)>0$.
For $g(\lambda_1,p)$, we take the partial derivative with respect to $p$: 
\[
\frac{\partial g}{\partial p} = \lambda_1^{p-1}\ln{\lambda_1}(\lambda_1-2),
\] 
which is non-negative for all $\lambda_1\leq 1$. Therefore $g(\lambda_1,p) \geq g(\lambda_1,0) = 2\lambda_1 - 2 > 0$.

Since both $g(\lambda_1,p)$ and $f(\lambda_2,\lambda_1)$ are strictly positive, we have $\frac{\partial u_2}{\partial \lambda_2}>0$. Therefore, the utility of agent $2$ is at least $u_2(\lambda_1,0)$. Plugging $\lambda_2 = 0$ in the definition of $u_2$ we have
\[
\frac{ -\lambda_1^{p + 1}  + 2 \lambda_1^p - \lambda_1
}{2\left(2 \lambda_1^p-\lambda_1\left(\lambda_1^p+1\right)\right)} = \frac{ -\lambda_1^{p}  + 2 \lambda_1^{p-1} - 1
}{2\left(2 \lambda_1^{p-1}-\lambda_1^p-1 \right)} = \frac{1}{2}.
\]

\paragraph{Case 3.}
There are exactly three positive variables, $v_1$, $v_2$ and $v_3$, and both the second and third constraints are tight (i.e. $v_1 + v_2 + v_3 = 1$ and $v_1 \lambda_1 + v_2 \lambda_2 + v_3 \lambda_3 = 1$). Since both agents have seen a value of $1$, and the first agent to reach a critical point was agent $1$, then by definition the utility of agent $2$ is at least $1/2$.

This concludes the proof of Lemma~\ref{lem: guarded is ef}.
\end{proof}
\subsection*{Instances of Figure~\ref{fig:cpplot}}
In the table below we present some of the instances that we used in the plot of Figure~\ref{fig:cpplot}. For all the instances where a critical point (CP) is not reached, we have two rounds, with $v_{11} = v_{22}$ and $v_{12} = 1-v_{11} = v_{21}$,therefore knowing $v_{11}$ is sufficient for constructing the whole instance. On the other hand, for the instances where a critical point is reached, we have a three rounds instance where $v_{12} = 1-v_{11}$, $v_{23} = 1-v_{21}$, and $v_{22} = v_{13}=\epsilon$. where $\epsilon$ is an arbitrarily small positive number. In this case, it is sufficient to know $v_{11}$ and $v_{21}$ to construct the whole instance.

In the table below, for each value of $p$ from the plot of Figure~\ref{fig:cpplot} we provide the instance that gives the lower of the two approximation upper bounds. For the case of $p=2.7$ we include both instances. And with rounding the two instances (approximately) coincide.\\

\begin{center}

\begin{tabular}{ c c c c }
  p value & with CP & Instance & Approx \\ 
  \hline
  $2$ & No & $v_{11}=v_{22} = 0.626$ & $0.894$ \\ 
  $2.1$ & No & $v_{11}=v_{22} = 0.621$ & $0.898$ \\
  $2.2$ & No & $v_{11}=v_{22} = 0.617$ & $0.902$\\
  $2.3$ &No & $v_{11}=v_{22} =0.613$ & $0.905$\\
  $2.4$ & No & $v_{11}=v_{22} = 0.609$ & $0.908$\\
  $2.5$ & No & $v_{11}=v_{22} = 0.606$ & $0.911$\\
  $2.6$ & No & $v_{11}=v_{22} = 0.602$ & $0.914$\\
  $\mathbf{2.7}$ & No &$v_{11} = v_{22}=0.599$ & $\mathbf{0.916}$\\
  $\mathbf{2.7}$ & Yes &$v_{11} =0.76$, $v_{22}=0.97$ & $\mathbf{0.916}$\\
  $2.8$ & Yes & $v_{11} = 0.75$, $ v_{21} = 0.96$& $0.912$\\
  $2.9$ & Yes &$v_{11} = 0.74$, $v_{21} = 0.95$ & $0.908$ \\
  $3.0$ & Yes & $v_{11} = 0.73$, $v_{21} = 0.94$ & $0.904$ \\
\end{tabular}

\end{center}

\subsection*{Missing from the proof of Theorem~\ref{thm: guarded poly approx}}

\begin{proof}[Proof of Theorem~\ref{thm: guarded poly approx} continued]

We continue the case analysis here.

\paragraph{Case 1.}
One item before, and two items after the critical point. Without loss of generality the first item is $1$ and the other two are $2$ and $3$. For convenience, normalize all values to add up to $1$.

Since $c$ is a critical point, the utility of agent $1$ for agent $2$'s allocation is $1/2$, or
\[
v_1 \frac{\lambda_1^p}{1 + \lambda_1^p} = \frac{1}{2} \implies v_1 = \frac{1 + \lambda_1^p}{2 \lambda_1^p}.
\]
The above equation further implies that, since $v_1 \leq 1$, then $\lambda_1 \geq 1$. This implies that in the optimal solution the first item goes to agent $2$, so the maximum possible welfare is $2-v_1$ (and this is tight, by setting $\lambda_1 = \frac{1}{v_1}$ and $v_2 = 1-v_1$).

Since agent $1$ gets utility $1/2$, the objective is
\[
\frac{1}{2} + v_1 \frac{\lambda_1^{p+1}}{1+\lambda_1^p} = \frac{1}{2} + \frac{1 + \lambda_1^p}{2 \lambda_1^p}\frac{\lambda_1^{p+1}}{1+\lambda_1^p} = \frac{1+\lambda_1}{2}.
\]
Therefore we have:
\[
\alpha = \frac{ALG}{OPT} = \frac{1+\lambda_1}{2(2-v_1)} = \frac{1+\lambda_1}{4-\frac{1+\lambda_1^{p}}{\lambda_1^{p}}} = \frac{\lambda_1^{p}+\lambda_1^{p+1}}{3\lambda_1^{p}-1}. 
\]

Taking the derivative we have
\[
\frac{\partial \alpha}{\partial \lambda_1} =
\frac{\lambda_1^{p-1}( 3\lambda_1^{p+1} - (1+p)\lambda_1  - p  ) }{(3\lambda_1^p-1)^2}.
\]

Solving $\frac{d \alpha}{d \lambda_1} = 0$ for $p=2.7$ we get that $\lambda_1 = 1.27764$; plugging it back in we have
\\

\includegraphics[width=0.9\linewidth]{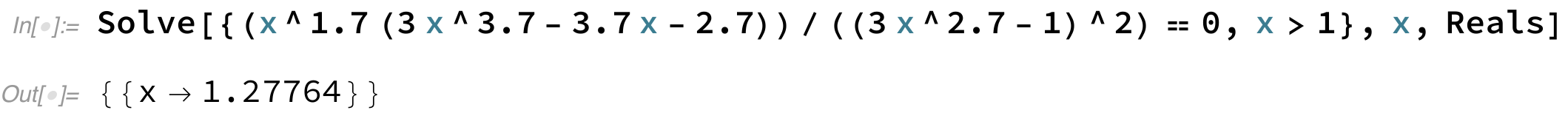}

\[\alpha(1.27764) > 0.916. \]
\paragraph{Case 2.} 
Two items before, and one item after the critical point. Without loss of generality the items are $v_1, v_2, v_3 >0$ and $c=2$. First notice that if cannot be that both $v_3 > 0$ and $\lambda_3 > 0$: otherwise we can construct a worse instance by splitting the third item into two rounds, where each agent wants a different item, thereby making the optimal welfare larger, but keeping the algorithm's welfare the same.
Furthermore, the case that agent $2$ has positive value for the item, but agent $1$ has zero value for the item is impossible (since $v_3 > 0$ by the original argument). Therefore, $\lambda_3 = 0$.

Since agent $2$ has seen all her value we get
\[
v_1 \lambda_1 + v_2 \lambda_2 = 1 \implies v_2 = \frac{1 - v_1 \lambda_1}{\lambda_2}.
\]
Since agent $1$ has value $1/2$ for agent $2$'s allocation we have
\[
v_1 \frac{\lambda_1^{p}}{1+\lambda_1^{p}}+\frac{1-v_1\lambda_1}{\lambda_2} \frac{\lambda_2^{p}}{1+\lambda_2^{p}} = \frac{1}{2}
\]
which implies
\begin{gather*}
   v_1 = \frac{(1+\lambda_2^p - 2\lambda_2^{p-1})(1+\lambda_1^p)}{2 \left( \lambda_1^p (1+\lambda_2^p) - \lambda_1 \lambda_2^{p-1} (1+\lambda_1^p)  \right) } \text{ , and } \\
   v_2 = \frac{1}{\lambda_2} - \frac{ \lambda_1 (1+\lambda_2^p - 2\lambda_2^{p-1})(1+\lambda_1^p)}{2 \lambda_2 \left( \lambda_1^p (1+\lambda_2^p) - \lambda_1 \lambda_2^{p-1} (1+\lambda_1^p)  \right) }.
\end{gather*}
The algorithm's welfare is
\[
ALG = \frac{1}{2} + v_1 \frac{\lambda_1^{p+1}}{1+\lambda_1^p} + v_2\frac{\lambda_2^{p+1}}{1+\lambda_2^p}.
\]
We can break cases based on the value of $\lambda_1$ and $\lambda_2$. They cannot both be strictly smaller than $1$ (otherwise we can't have both $v_1\lambda_1 + v_2\lambda_2 = 1$ and $v_1+v_2+v_3=1$). The case that both equal to $1$ is trivial (identical agents), and the $\lambda_1 > 1, \lambda_2 < 1$ and $\lambda_1 < 1, \lambda_2 > 1$ cases are symmetric, so we only need to consider one of them. First, without loss of generality assume that $\lambda_1 > 1, \lambda_2 < 1$. For convenience we normalize all values to add up to $1$.
The optimal welfare is $OPT = 2 - v_1 - v_2\lambda_2$. We can therefore write
\[
\alpha(\lambda_1,\lambda_2) = \frac{ALG}{OPT} = \frac{\frac{1}{2} + v_1 \frac{\lambda_1^{p+1}}{1+\lambda_1^p} + v_2\frac{\lambda_2^{p+1}}{1+\lambda_2^p}}{2 - v_1 - v_2\lambda_2}, 
\]

Let $\lambda_1 = x$, $\lambda_2 = y$, and plug in the closed forms for $v_1$ and $v_2$. \\

\includegraphics[width=0.9\linewidth]{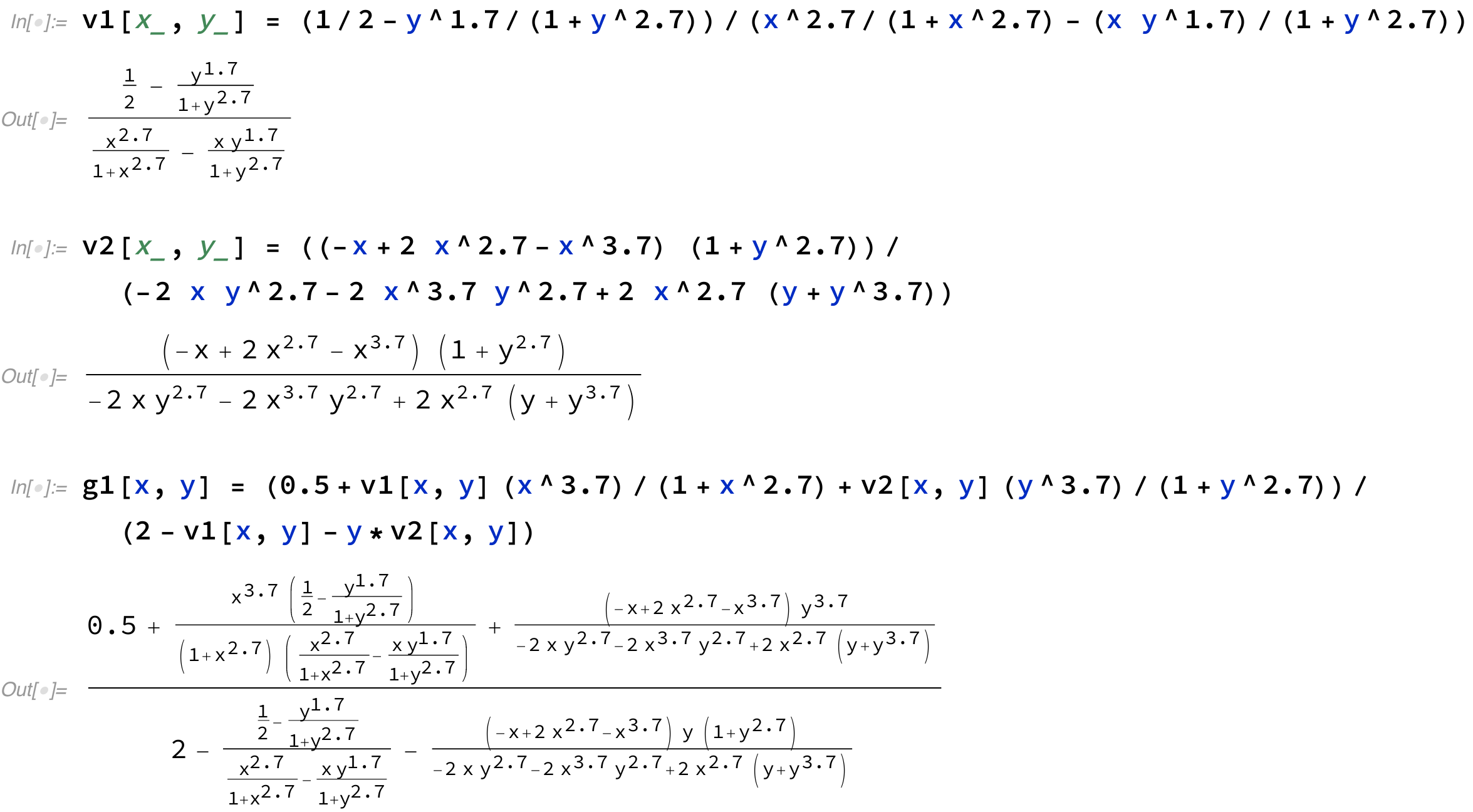}\\

Minimizing $\alpha$ in the feasible region we have\\

\includegraphics[width=\linewidth]{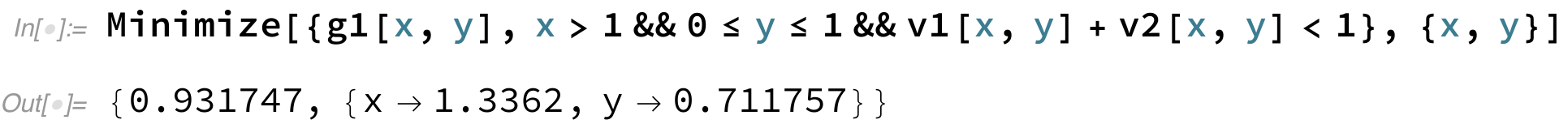}\\
\[
\alpha(1.3362,0.711757) > 0.93.
\]

Now consider the case where $\lambda_1$, $\lambda_2>1$. The optimal welfare is now $OPT = 2-v_1-v_2$, therefore

\[\alpha(\lambda_1,\lambda_2) = \frac{ALG}{OPT} = \frac{\frac{1}{2} + v_1 \frac{\lambda_1^{p+1}}{1+\lambda_1^p} + v_2\frac{\lambda_2^{p+1}}{1+\lambda_2^p}}{2 - v_1 - v_2}.\]

Using Mathematica we have \\

\includegraphics[width=0.9\linewidth]{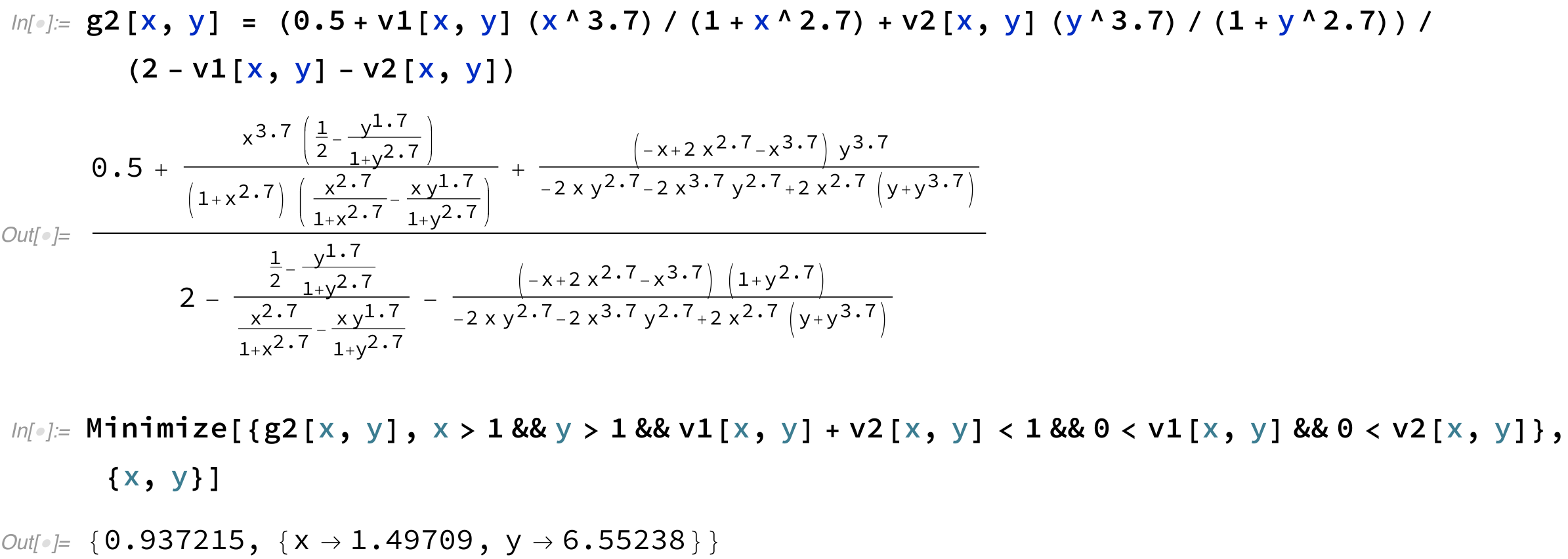}

\[\alpha(1.49709, 6.55238) > 0.93\]
\paragraph{Case 3.} 
Three items before $c$, i.e. $c = 3$, and all items after the critical point have zero value for both agents. Then, by definition, there is no critical point. From the analysis of the (un-guarded) poly-proportional algorithm (Claim~\ref{claim: poly prop worst case}), we already know that the worst approximation is achieved by a two item instance.

Consider the instance (in standard notation)
\begin{center}
\begin{tabular}{ c c } 
  Agent one & Agent two \\ 
  \hline
  $1-v_1$ & $v_2$ \\ 
  $v_1$ & $1-v_2$ \
\end{tabular}
\end{center}

Without loss of generality assume that $v_2 > 1-v_1$ and $v_1 > 1-v_2$. The optimal welfare is $OPT = v_1 + v_2$.

Now consider the the performance of algorithm:
\[
ALG = \frac{(1-v_1)^{p+1}+v_2^{p+1}}{(1-v_1)^{p}+v_2^{p}} + \frac{(1-v_2)^{p+1}+v_1^{p+1}}{(1-v_2)^{p}+v_1^{p}}.
\]
The approximation to optimal welfare is
\[\alpha = \frac{ALG}{OPT} = \frac{\frac{(1-v_1)^{p+1}+v_2^{p+1}}{(1-v_1)^{p}+v_2^{p}} + \frac{(1-v_2)^{p+1}+v_1^{p+1}}{(1-v_2)^{p}+v_1^{p}}}{v_1 + v_2}.\]

Let $v_1=x$ and $v_2=y$. Using Mathematica in order to minimize $\alpha$ in the feasible region we have\\

\includegraphics[width=\linewidth]{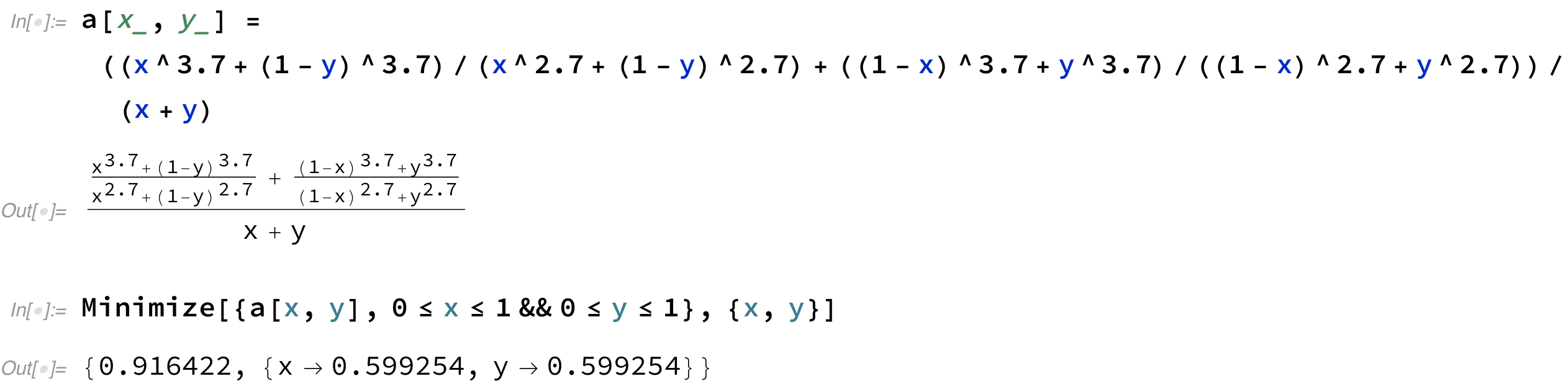}

\begin{align*}
\alpha(0.599,0.599) &=\frac{\frac{0.401^{3.7}+0.599^{3.7}}{(0.401)^{2.7}+0.599^{2.7}} + \frac{0.401^{3.7}+0.599^{3.7}}{0.401^{2.7}+0.599^{2.7}}}{1.198}\\
&=0.9164 > 91.6.
\end{align*}
This concludes the proof of Theorem~\ref{thm: guarded poly approx}.
\end{proof}

\end{document}